\documentclass[12pt]{article}

\usepackage{amsmath}
\usepackage{graphicx}
\usepackage{framed, color} 
\usepackage[margin=.7in]{geometry}
\usepackage{pdflscape}
\definecolor{shadecolor}{rgb}{1, 0.8, 0.3}
\usepackage{epstopdf}

\usepackage{amssymb}
\usepackage{amsthm}

\usepackage{amsfonts}
\usepackage{cite}
\usepackage{ifthen}
\usepackage{epsfig}
\usepackage{subcaption}
\usepackage{array}
\usepackage{enumerate}
\usepackage{algorithm}
\usepackage[noend]{algpseudocode} 
\usepackage[all]{xy}
\usepackage{algpseudocode}
\usepackage[compact]{titlesec}
\usepackage{url}

\usepackage{times}
\usepackage{tikz}
\usepackage{float}
\usetikzlibrary{shapes,positioning,arrows,automata}
\usepackage{caption}
\newtheorem{theorem}{Theorem}[section]
\newtheorem{lemma}[theorem]{Lemma}
\newtheorem{proposition}[theorem]{Proposition}
\newtheorem{corollary}[theorem]{Corollary}

\newcommand{\bi}{\ensuremath{\mathbf{i}}}
\newcommand{\om}{\ensuremath{\omega}}

\newcommand{\PP}{\ensuremath{\mathbb{P}}}
\newcommand{\bd}{\ensuremath{\bar{d}}}
\newcommand{\ty}{\ensuremath{\tilde{y}}}

\newcommand{\beq}{\begin{equation}}
\newcommand{\eeq}{\end{equation}}
\newcommand{\bea}{\begin{eqnarray}}
\newcommand{\eea}{\end{eqnarray}}
\newcommand{\bean}{\begin{eqnarray*}}
\newcommand{\eean}{\end{eqnarray*}}
\newcommand{\bit}{\begin{itemize}}
\newcommand{\eit}{\end{itemize}}
\newcommand{\ben}{\begin{enumerate}}
\newcommand{\een}{\end{enumerate}}
\newcommand{\blem}{\begin{lem}}
\newcommand{\elem}{\end{lem}}
\newcommand{\bthm}{\begin{thm}}
\newcommand{\ethm}{\end{thm}}
\newcommand{\bpf}{\begin{IEEEproof}}
\newcommand{\epf}{\end{IEEEproof}}

\newcommand{\comment}[1]{}

\newcommand{\algorithmfootnote}[2][\footnotesize]{%
  \let\old@algocf@finish\@algocf@finish
  \def\@algocf@finish{\old@algocf@finish
    \leavevmode\rlap{\begin{minipage}{\linewidth}
    #1#2
    \end{minipage}}%
  }%
}

\begin{document}

\title{Capacity-Approaching \emph{PhaseCode} for \\ Low-Complexity Compressive Phase Retrieval}
\author{Ramtin Pedarsani, Kangwook Lee, and Kannan Ramchandran\\ Dept. of Electrical Engineering and Computer Sciences\\ University of California, Berkeley\\ \{ramtin, kw1jjang, kannanr\}@eecs.berkeley.edu}
\date{}
\maketitle

\begin{abstract}
In this paper, we tackle the general compressive phase retrieval problem. The problem is to recover (to within a global phase uncertainty) a $K$-sparse complex vector of length $n$, $x \in \mathbb{C}^n$, from the magnitudes of $m$ linear measurements, $y = |Ax|$, where $A \in \mathbb{C}^{m \times n}$ can be designed, and the magnitudes are taken component-wise for vector $Ax \in \mathbb{C}^m$. We propose a variant of the PhaseCode algorithm, first introduced in \cite{PLR14}, and show that, using an irregular left-degree sparse-graph code construction, the algorithm can recover almost all the $K$ non-zero signal components using only slightly more than $4K$ measurements under some mild assumptions, with order-optimal time and memory complexity of $\mathcal{O}(K)$. It is known that the fundamental limit for the number of measurements in compressive phase retrieval problem is $4K - o(K)$ \cite{Tarokh,Heinosaari}. To the best of our knowledge, this is the first constructive \emph{capacity-approaching} compressive phase retrieval algorithm: in fact, our algorithm is also order-optimal in complexity and memory.

As a  second contribution of this paper, we propose another variant of the PhaseCode algorithm that is based on a Compressive Sensing  framework involving sparse-graph codes.  Our proposed algorithm is an instance of a more powerful
``separation" architecture that can be used to address the compressive phase-retrieval problem in general.  This modular design features a compressive sensing outer layer, and a trigonometric-based phase-retrieval inner layer.  The compressive-sensing layer
 operates as a linear phase-aware compressive measurement subsystem, while the trig-based phase-retrieval layer
 provides the desired abstraction between the actually targeted nonlinear phase-retrieval problem  and the induced linear compressive-sensing problem.  Invoking this architecture based on the use of sparse-graph codes for the compressive sensing layer, we show that we can exactly recover a signal, to within a global phase, from only the magnitudes of its linear measurements using only slightly more than $6K$ measurements. 
\end{abstract}

\section{Introduction}\label{sec:intro}
The general compressive phase retrieval problem is to recover a $K$-sparse signal of length $n$, $x \in \mathbb{C}^n$, to within a global phase uncertainty, from the magnitude of $m$ linear measurements $y = |Ax|$, where $A \in \mathbb{C}^{m \times n}$ is the measurement matrix that can be designed, and the magnitude is taken on each component of the vector $Ax \in \mathbb{C}^m$.\footnote{This can be easily extended to the case where the signal $x$ is sparse with respect to some other basis, for example in Fourier domain, by a simple transformation of the measurement matrix.} In many applications such as  optics \cite{Walther}, X-ray crystallography \cite{Milane,Harrison}, astronomy \cite{Dainty}, ptychography \cite{Rodenburg}, quantum optics \cite{Mohammad}, etc., the phase of the measurements is not available. Furthermore, in many applications the signal of interest is sparse in some domain. For example, compressive phase retrieval has been used in a recent work in quantum optics\cite{Mohammad} to measure the transverse wavefunction of a photon.

The phase retrieval problem has been studied extensively in the literature.  Here, we only provide a brief review of related works on compressive phase retrieval, and refer the readers to our earlier work \cite{PLR14} for a more extensive literature review. 

To the best of our knowledge, the first algorithm for compressive phase retrieval was proposed by Moravec \emph{et al.} in \cite{Baraniuk}. This approach requires knowledge of the $\ell_1$ norm of the signal, making it impractical in most scenarios. The authors in \cite{Tarokh} showed that $4K-1$ measurements are theoretically sufficient to reconstruct the signal, but did not propose any low-complexity algorithm for reconstruction. The ``PhaseLift'' method originally proposed for the non-sparse case \cite{Candes1}, is also used for the sparse case in \cite{Sastry} and \cite{Li}, requiring $\mathcal{O}(K^2\log(n))$ intensity measurements, and having a computational complexity of $\mathcal{O}(n^3)$, making the method less practical for large-scale applications.   
An alternating minimization method is proposed in \cite{Sanghavi} for the sparse case with $\mathcal{O}(K^2\log(n))$ measurements and a computation complexity of $\mathcal{O}(K^3n \log(n))$.  
Compressive phase-retrieval  via generalized approximate message passing (PR-GAMP) is proposed in \cite{Rangan}, with good performance in both runtime and noise robustness shown via simulations. Jaganathan \emph{et al.} consider the phase retrieval problem from Fourier measurements \cite{Hassibi1,Hassibi2}. They propose an SDP-based algorithm, and show that the signal can be provably recovered with $\mathcal{O}(K^2 \log(n))$ Fourier measurements \cite{Hassibi1}. Cai \emph{et al.} propose proposed SUPER algorithm for compressive phase-retrieval in \cite{Jaggi}. The SUPER algorithm uses  $\mathcal{O}(K)$ measurements and features $\mathcal{O}(K\log(K))$ complexity with zero error floor asymptotically. Recently, Yapar \emph{et al.} propose a fast compressive phase retrieval algorithm using $4K\log(n/K)$ random Fourier measurements based on an algebraic phase retrieval stage and a compressive sensing step subsequent to it \cite{yapar}.

\subsection{Main Contribution}\label{sec:main}
The main contribution of this paper is to propose two variants of PhaseCode algorithm, which was first introduced in \cite{PLR14}. The measurement matrix in PhaseCode algorithm is constructed based on a {\em {sparse-graph code}}.  The work of \cite{PLR14}, which features a regular left-degree distribution in its sparse-graph code design, establishes that, with high probability, all but a fraction $10^{-7}$ (error floor) of the non-zero signal components can be recovered using about $14K$ measurements with optimal time and memory complexity.  

In this work, we first consider an {\em irregular} left-degree distribution for the sparse-graph code construction. We show that under this construction, the PhaseCode algorithm can recover almost all the non-zero components of the signal using only slightly more than $4K$ measurements under some mild assumptions, with optimal time and memory complexity of $\mathcal{O}(K)$.  It is well-known that $4K - o(K)$ measurements is the fundamental limit for unique recovery of $K$-sparse signal $x$ \cite{Tarokh,Heinosaari}. This shows that the irregular PhaseCode algorithm is \emph{capacity-approaching}.  We note a few caveats, which are why we call our algorithm capacity-approaching.  First, our irregular PhaseCode has an arbitrarily small error floor, which is however still non-vanishing as $K$ goes to infinity.  Secondly, in order to prove that the algorithm can exactly recover $x$, we need to make some mild assumptions on $x$.

Second, we propose a different two-layer PhaseCode algorithm: a compressive sensing layer and a phase retrieval layer.\footnote{While we were writing this paper, we became aware that a similar approach was used in \cite{yapar} in a contemporaneous work. Although
the high-level modularity concepts are similar, the proposed algorithms are significantly different.  Principally, our measurement matrix is designed using sparse-graph codes, making our algorithm distinct from that of \cite{yapar}.}  Our proposed algorithm 
is an instance of a more powerful
``separation" architecture that can be used to address the compressive phase-retrieval problem in general.  This modular design features a compressive sensing outer layer, and a trigonometric-based phase-retrieval inner layer.  The compressive-sensing layer
 operates as a linear phase-aware compressive measurement subsystem, while the trig-based phase-retrieval layer
 provides the desired abstraction between the targeted nonlinear phase-retrieval problem  and the induced linear compressive-sensing problem.  In the compressive sensing layer of the algorithm, we use the recent sparse-graph code based framework of \cite{Simon}, to design only $m_1 \simeq 2K$ measurements that guarantee perfect recovery for the compressive sensing problem, and in the phase-retrieval layer of the algorithm, we use $3m_1 -2$ deterministic measurements proposed in \cite{PLR14} that recover the phase of the $m_1$ measurements in phase and magnitude. Thus, we show that the signal can be reconstructed perfectly using only slightly more than $6K$ measurements.    

\section{Problem definition}
Consider a $K$-sparse complex signal $x \in \mathbb{C}^n$ of length $n$. Define $A \in \mathbb{C}^{m \times n}$ to be the measurement matrix that needs to be designed. The phase retrieval problem is to recover the signal $x$ from magnitude of linear measurements $y_i = |a_i x|$, where $a_i$ is the $i$-th row of matrix $A$. Figure \ref{fig:formulation} illustrates the block diagram of our problem. 

\begin{figure}[h]
\centering
    \includegraphics[width= 0.9\textwidth]{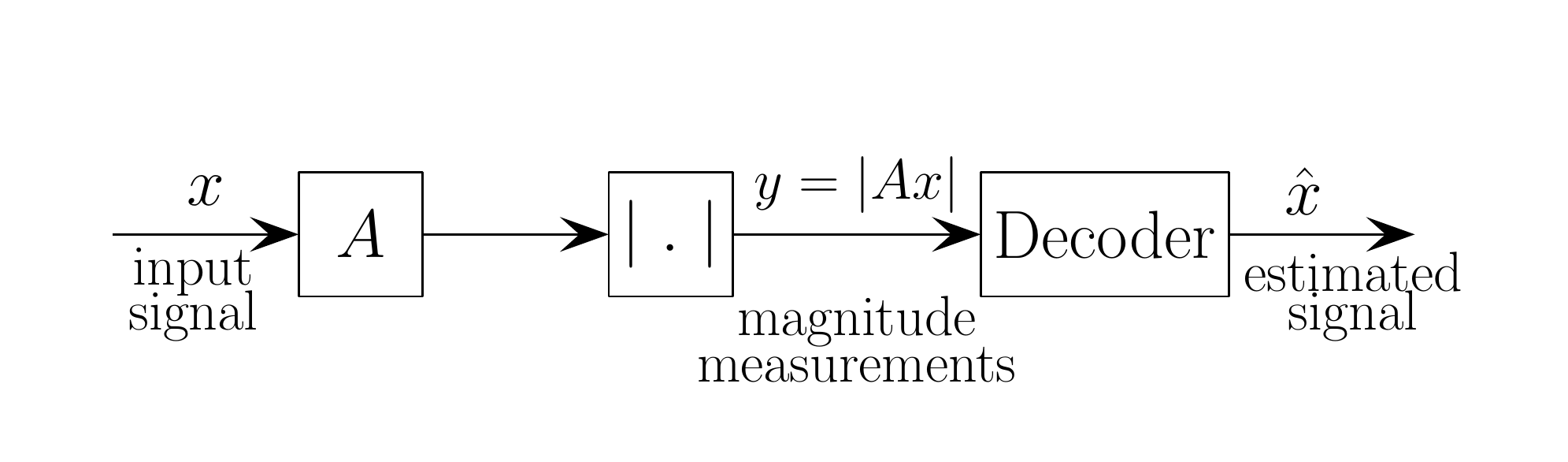}
  \caption{Block diagram of the general compressive phase retrieval problem. The measurements are $y_i = |a_i x|$, where $a_i$ is the $i$-th row of measurement matrix $A$. The objective is to design measurement matrix $A$ and the decoding algorithm to guarantee high reliability, while having small sample complexity as well as small time and memory complexity. \label{fig:formulation}} 
\end{figure}

The main objectives of the general compressive phase retrieval problem are to design matrix $A$, and the decoding algorithm to recover $x$, that satisfy the following objectives. 
\begin{itemize}
\item The number of measurements $m$ is as small as possible. Ideally, one wants $m$ to be close to the fundamental limit of $4K - o(K)$ \cite{Tarokh,Heinosaari}. 
\item The decoding algorithm is fast with low computational complexity and memory requirements. Ideally, one wants the time complexity and the memory complexity of the algorithm to be $\mathcal{O}(K)$, which is optimal. 
\item The reliability of the recovery algorithm should be maximized. Ideally, one wants the probability of failure to be vanishing as the problem parameters $K$ and $m$ get large.
\end{itemize}

\section{PhaseCode based on irregular left-degree sparse-graph codes}
We consider the same architecture used in \cite{PLR14} to design our measurement matrix, $A$. We define $A \in \mathbb{C}^{pM \times n}$ to be a row tensor product of a trigonometric modulation matrix, $T \in \mathbb{C}^{p \times n}$, and a code matrix, $H \in \mathbb{C}^{M \times n}$. The row tensor product of matrices $T$ and $H$ is defined as follows. If $A = T \otimes H = [A_1^T, A_2^T, \ldots, A_M^T]^T$ and $A_i \in \mathbb{C}^{p \times n}$, then, $A_i(jk) = T_{jk}H_{ik}, ~ 1\leq j \leq p, ~ 1 \leq k \leq n$. As an example, the row tensor product of matrices 
$$
T = \left [ \begin{array}{ccc}
0.1 & 0.2 & 0.3 \\
0.4 & 0.5 & 0.6 
\end{array} \right] 
 ~ \text{and} ~ 
H = \left [ \begin{array}{ccc}
0 & 1 & 0 \\
1 & 1 & 0 \\
0 & 0 & 1
\end{array} \right] 
$$
is 
$$
A = T \otimes H = \left [ \begin{array}{ccc}
0 & 0.2 & 0 \\
0 & 0.5 & 0 \\
0.1 & 0.2 & 0 \\
0.4 & 0.5 & 0 \\
0 & 0 & 0.3 \\
0 & 0 & 0.6
\end{array} \right].
$$ 
Our trigonometric modulation matrix is the same as the one we proposed in \cite{PLR14}; that is,
\begin{equation}\label{eq:T}
T = \left( \begin{array}{cccc}
e^{\bi \om} & e^{\bi 2 \om} & \ldots & e^{\bi n \om} \\
e^{-\bi \om} & e^{-\bi 2 \om} & \ldots & e^{-\bi n \om} \\
\cos(\om) & \cos(2\om) & \ldots & \cos(n \om) \\
e^{\bi \om '} & e^{\bi 2 \om '} & \ldots & e^{\bi n \om '}
\end{array}
\right),
\end{equation}
where $\om = \frac{2 \pi}{n}$ and $\om'$ is a random phase uniformly distributed in $[0,2\pi)$. 

\begin{figure}
\centering
    \includegraphics[width= 0.3\textwidth]{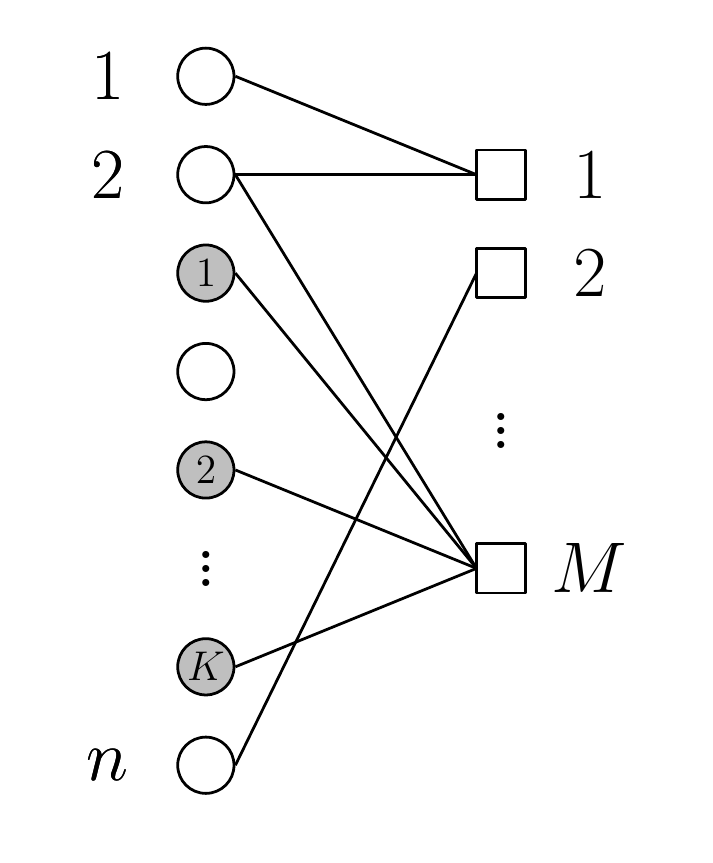}
  \caption{Illustration of a bipartite graph representation of the compressive phase-retrieval problem. In this graph, each left node refers to a signal component and each right node refers to a set of (four) measurements. The signal is $K$-sparse, so $K$ of these left nodes are active as illustrated in the figure. The four measurements associated with each right node correspond to the trigonometric-based modulation measurements of Equation (\ref{eq:T}).\label{fig:graph}} 
\end{figure}

Now we recall the key idea of the PhaseCode algorithm.\footnote{In this paper, we consider only the Unicolor PhaseCode algorithm proposed in \cite{PLR14}.} For more details, we refer the readers to \cite{PLR14}. The architecture of the PhaseCode algorithm is as follows. We consider a bipartite graph of $n$ left nodes and $M$ right nodes, where each left node is a signal component and each right node is a set of $4$ measurements. One can think of the $n$ left nodes as $n$ slots in which $K$ balls will be thrown arbitrarily. The $n$ slots refer to the indices of  the $n$-length signal, and the $K$ balls refer to the $K$ non-zero components of the signal. See Figure \ref{fig:graph} as an illustration.  Our goal is to design the bipartite graph such that the $K$ locations (belonging to the integer set between $1$ and $n$) of these balls (non-zero components) are detected. Furthermore, the relative phase and magnitude of the non-zero components need to be detected with the aid of the trigonometric modulation matrix. (See Subsection \ref{sec:trig} and \cite{PLR14} for details.)  

In order to explain our design for the phase-retrieval problem, it is illustrative for us to consider an equivalent {\em ball-coloring} 
problem in the following sense.  The goal of this equivalent problem is to detect and color all the balls in the system, where we say that a ball is colored if it is detected that it corresponds to an active signal component in that slot, and it is recovered in magnitude and phase (upto a global phase).  Thus, when balls get colored, their locations and complex-amplitudes are detected, where we may think of the color as the global coordinate-system on which we describe the unknown signal components.  Thus, the ball-coloring problem is equivalent to the targeted phase-retrieval problem.  However, the {\em rules} needed to color these balls are not yet clear, and need to be specified.  This is where our carefully-designed trignometric-modulated measurements of Eq. (\ref{eq:T}) come in.  Concretely, in our proposed construction, these trig-modulated measurements are used to enforce the following coloring rule.   \emph{If a right node of the bipartite graph is connected to $n' \geq 2$ balls and $n'-1$ of those balls are colored, the single uncolored ball can get colored.} See Figure \ref{fig:multiton} as an illustration.  What this means for our targeted phase-retrieval problem is that if we have already uncovered the locations and complex amplitudes (including phase) of $n'-1$ of the $n'$ signal components corresponding to a right node, then the trig-modulated measurements can be used to uncover the lone missing signal component; i.e., it too can be colored.   Of course, at the outset of the algorithm, no balls are colored. Therefore, the algorithm requires an initial seeding phase which will jumpstart the process by ensuring that a few balls get colored. (We will show how this can be realized in Subsection \ref{sec:initial}.) Given that in the initial phase of the algorithm, there are only a few colored balls in the graph, the goal is to design a bipartite graph such that the coloring spreads like an epidemic until all the balls get colored.   Given this exact equivalence to our original phase-retrieval problem, in the interests of illustrative clarity,
from now on, we mainly focus on the ball coloring problem, and connect it back to the phase retrieval problem in Subsection \ref{sec:trig}.

\begin{figure}
        \centering        
                \includegraphics[width=0.35\textwidth]{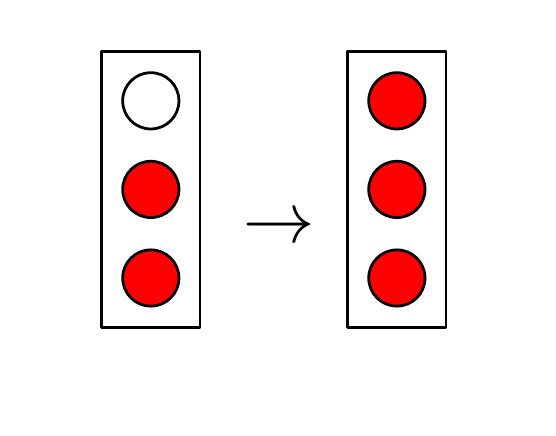}
        \caption{This figure illustrates when a bin contains one or more colored balls and {\em exactly one} uncolored ball, then the uncolored ball gets colored.\label{fig:multiton}}
\end{figure}

As mentioned earlier, we design our code matrix based on a random bipartite graph model. Define $\rho_i$ to be the probability that a randomly selected edge is connected to a right node of degree $i$, and $\lambda_i$ to be the probability that a randomly selected edge is connected to a left node of degree $i$. Define the edge degree distributions or edge degree polynomials of right and left nodes as follows.
\begin{align*}
\rho(x) &= \sum_{i \geq 1} \rho_i x^{i-1}; \\
\lambda(x) &= \sum_{i \geq 1} \lambda_i x^{i-1}.
\end{align*}

To analyze the PhaseCode algorithm, we use density evolution, which is a technique to analyze the performance of message passing algorithms on sparse-graph codes. We find a recursion relating the probability that a randomly chosen ball or left node in the graph is not colored after $j$ iterations of the algorithm, $p_{j}$, to the same probability after $j+1$ iterations of the algorithm, $p_{j+1}$.  Our analysis methodology is similar to that  of \cite{PLR14} with the distinction that we now consider a different left-degree distribution.  We first use a  tree-like assumption of the graph (i.e. the graph is free from cycles) and undertake a mean analysis using density evolution.  Then, we relax the tree-like assumption, as well as find concentration bounds to show that the actual performance concentrates around the mean.  Details follow.  First, under the tree-like assumption of the graph, one has
\begin{equation}\label{eq:density}
p_{j+1} = \lambda(1 + \rho_1 - \rho(1 - p_j)).
\end{equation}
Here is a proof of Equation \eqref{eq:density}. A right node (bin) passes a message to a left node (ball) that the left node can be colored if all of its neighbors are colored. This happens with probability 
\begin{align*}
\sum_{i = 2}^\infty \rho_i (1-p_j)^{i-1} = \rho(1-p_j) - \rho_1 
\end{align*}

\begin{figure}
\centering
    \includegraphics[width= 0.6\textwidth]{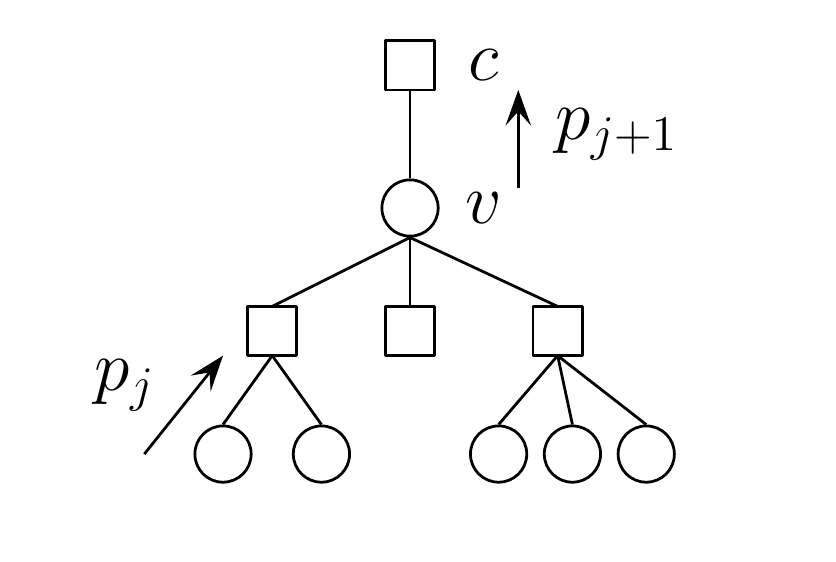}
  \caption{Illustration of message passing algorithm and density evolution equation. \label{fig:tree}} 
\end{figure}

Thus, the probability that ball $v$ passes a message to bin $c$ that it is not colored in iteration $j+1$, is the probability that none of its other neighbor bins can tell $v$ that it is colored in iteration $j$. See Figure \ref{fig:tree} for an illustration. These messages are all independent if the bipartite graph is a tree. Assuming this, one has 
\begin{align}
p_{j+1} &= \sum_{i \geq 1} \lambda_i [1 - (\rho(1-p_j) - \rho_1)]^{i-1} \nonumber \\
&= \lambda(1 + \rho_1 - \rho(1 - p_j)). 
\end{align} 

Recall that our goal is to design a bipartite graph such that all the balls get colored with as few measurements as possible. 
The way that we design the bipartite graph is to connect the left nodes to right nodes uniformly at random according to some distribution. Then, for large $K$ and $M$, the induced right node degree distribution is Poisson with parameter $\eta = \frac{K \bd}{M}$, where $\bd$ is the average degree of left nodes. Thus, 

\begin{align*}
\rho_i = \frac{iM}{K\bd} \PP(\text{random right node has degree}~i) 
 = \frac{i}{\eta} \frac{\eta^i e^{-\eta}}{i!} = \frac{\eta^{i-1}e^{-\eta}}{(i-1)!}.
\end{align*}
Then,
\begin{align} \label{eq:edge}
\rho(x) = \sum_{i \geq 1} \frac{\eta^{i-1}e^{-\eta}}{(i-1)!} x^{i-1} 
 = \sum_{i \geq 0} \frac{\eta^{i}e^{-\eta}}{i!} x^i 
= e^{-\eta} e^{\eta x}  
= e^{-\eta(1-x)}.
\end{align} 

We design the left degree distribution $\lambda(x)$ based on a truncated harmonic distribution as follows. 
Let $h(x) = \sum_{i=1}^x 1/i$. Then,
\begin{align}
\lambda_i = \frac{1}{i-1} \times \frac{1}{h(D-1)}, ~ 2 \leq i \leq D,
\end{align}
where $D$ is a (large) constant.

We design the number of right nodes to be $M = K/(1-\epsilon)\simeq K(1+\epsilon)$. How to choose constants $D$ and $\epsilon$ will be shortly clarified in Lemma \ref{lem:ir}. The average degree of left nodes is $\bd = \frac{1}{\sum_i \lambda_i/i}$ \cite{Luby2}. To see this, let $E$ be the number of edges of the graph. Then, the number of left nodes of degree $i$ is $E\lambda_i/i$ since $\lambda_i$ is the fraction of edges with degree $i$ on the left. Thus, the number of left nodes is $\sum_i E\lambda_i/i$. So the average left degree is 
$$
\bar{d} = \frac{E}{\sum_i E\lambda_i/i} = \frac{1}{\sum_i \lambda_i/i}.
$$ 
Thus, with our design, 
$$
\bar{d} = (\sum_{i=2}^D \frac{\lambda_i}{i})^{-1} = h(D-1) \frac{D}{D-1}.
$$
Consequently, the Poisson density parameter of the right-node degree distribution is:
$$
\eta = \frac{K\bd}{M} = h(D-1) \frac{D}{D-1} (1-\epsilon).
$$

\begin{figure}
        \centering
        \begin{subfigure}{0.35\textwidth}
                \includegraphics[width=7.5cm,height=5cm]{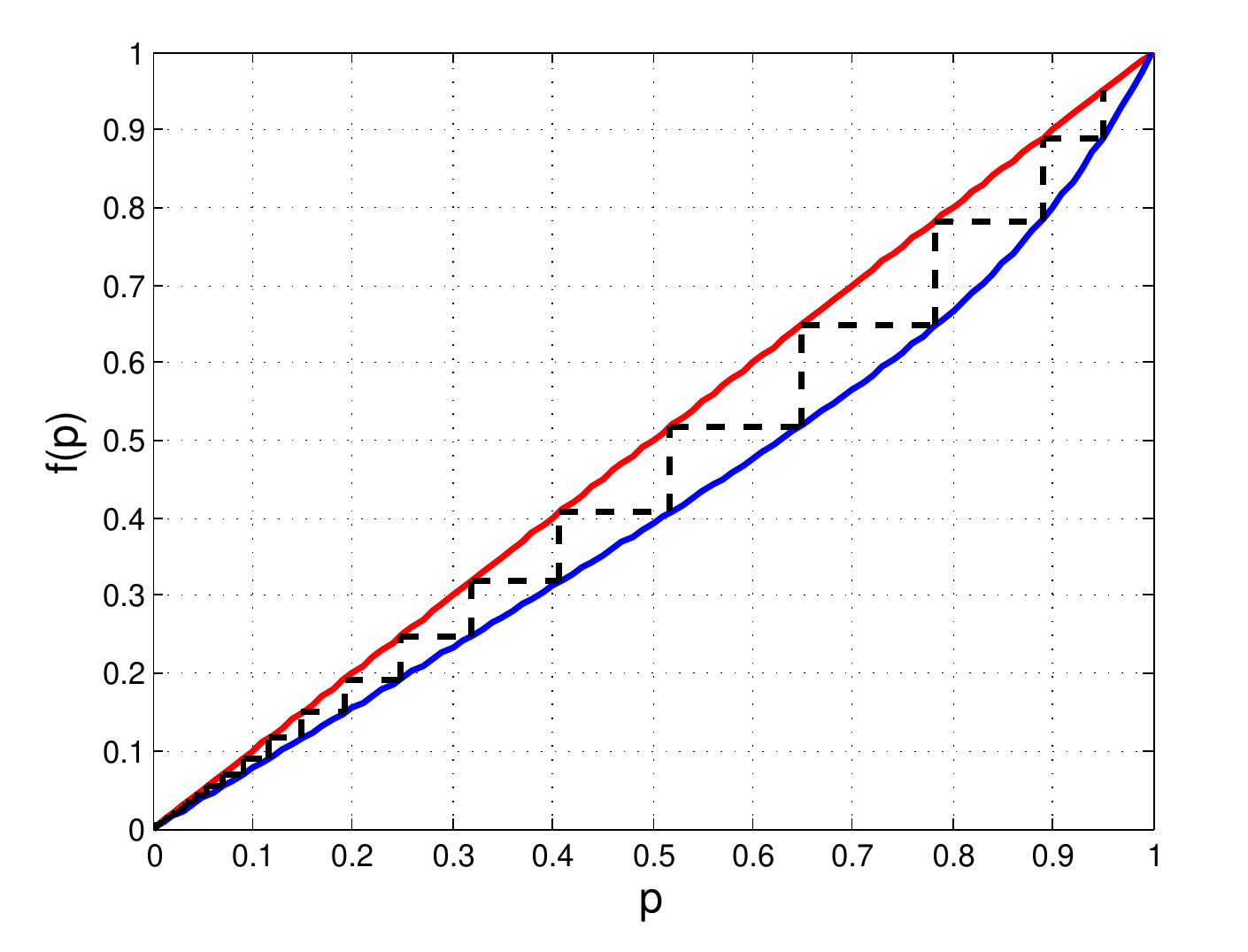}
                \caption{The density evolution curve for parameters $K = 10^5$, $\epsilon=0.3$ and $D = 10^3$.}
                \label{fig:de3}
        \end{subfigure}
        \qquad
        \begin{subfigure}{0.35\textwidth}
                \includegraphics[width=7.5cm,height=5cm]{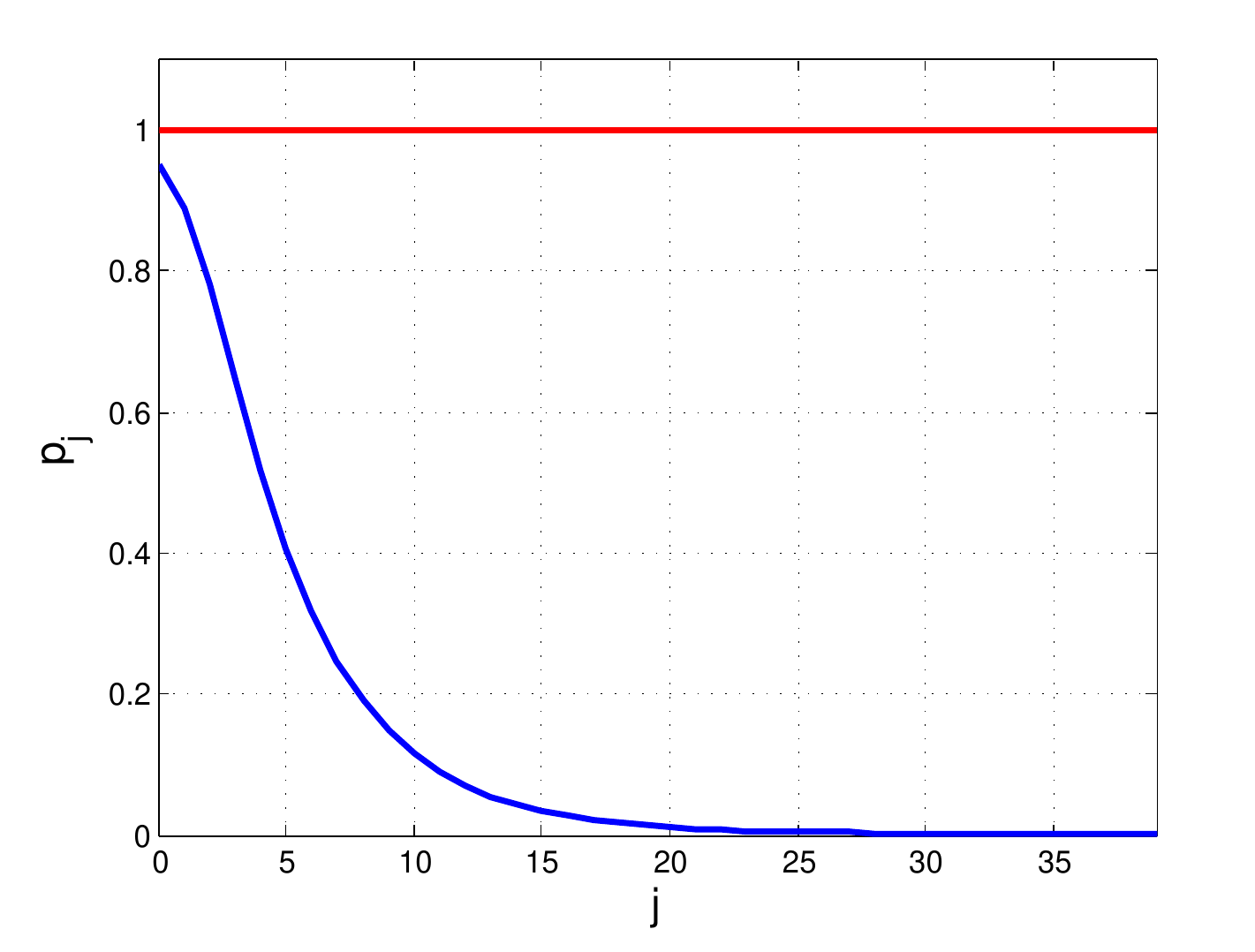}
                \caption{The evolution of $p_j$ after each iteration for parameters $K = 10^5$, $\epsilon=0.3$ and $D = 10^3$.}
                \label{fig:de4}
        \end{subfigure}
        \caption{Figure $(a)$ illustrates the density evolution equation: $p_{j+1} = f(p_j)$. In order to track the evolution of $p_{j}$, pictorially, one draws a vertical line from $(p_j,p_j)$ to $(p_j,f(p_j))$, and then a horizontal line between $(p_j,f(p_j))$ and $(f(p_j),f(p_j))$. Since the two curves meet at $(1,1)$ if $p_0 =1$, then $p_j$ gets stuck at $1$. However, if $p_0 = 1 - \delta$, for some $\delta > 0$, $p_j$ decreases after each iteration, and it gets very close to $0$ in finite number of iterations. Figure $(b)$ illustrates the same phenomenon by showing the evolution of $p_j$ versus the iteration, $j$. Note that in this example, since $\epsilon = 0.3$ and we are operating further away from the capacity, $p_j$ gets very close to $0$ after only $25$ iterations.}\label{fig:density2}
\end{figure}

\begin{lemma}\label{lem:de}
Let $f(x) = \lambda(1 + e^{-\eta} - e^{-\eta x})$. The fixed point equation $x = f(x)$ has exactly two solutions, $x^*_1 = 1$ and $0 < x^*_2 < 1$, in the interval $x \in [0,1]$. Furthermore, if $f'(1) > 1$, then $f(x) < x$ for $x \in (x^*_2,1)$. 
\end{lemma}

\begin{proof}
First note that it is easy to prove the lemma for specific parameters by plotting the function. See for example Figure \ref{fig:de1}. To formally show it, note that $f(1) = 1$ is one solution of the fixed point equation, since $\lambda(1) = 1$. Also $f(0) = \lambda(e^{-\eta}) > 0$. Thus, by continuity of $f(x)$ and using the assumption that $f'(1) > 1$, there is another fixed point $x^*_2$. Now since $f'(1) > 1$, $f(x) < x$ for $x$ close to 1. In order to show that $f(x) < x$ for all $x \in (x^*_2,1)$, it is enough to show that $f'(x) - 1 = 0 $ has only one solution in $x \in (0,1)$. To this end, see that
$$
f'(x) = \eta e^{-\eta x} \lambda'(1 + e^{-\eta} - e^{-\eta x}).
$$
For ease of notation, let $y = 1 + e^{-\eta} - e^{-\eta x}$ and $y \in (e^{-\eta},1)$. Equivalently, we want to show that 
$$
C(1 + e^{-\eta} - y)(1 + y + y^2 + \ldots + y^{D-2}) =1 
$$
has only one solution where $C = \eta/h(D-1)$. This is easy to see since $D$ is large so $y \simeq  \frac{1-C-Ce^{-\eta}}{1-C}$. 
\end{proof}

\begin{figure}
        \centering
        \begin{subfigure}{0.35\textwidth}
                \includegraphics[width=7.5cm,height=5cm]{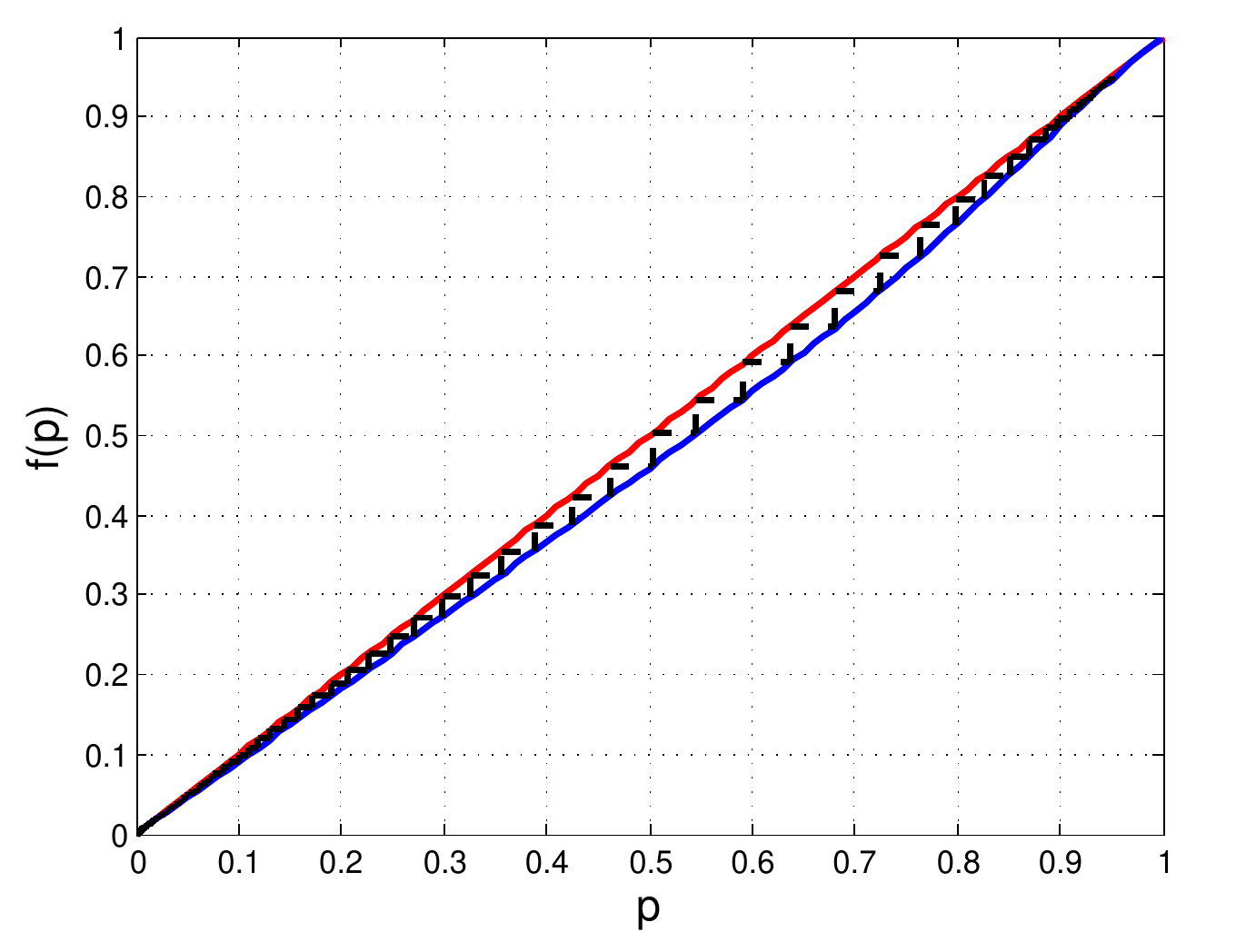}
                \caption{The density evolution curve for parameters $K = 10^5$, $\epsilon=0.1$ and $D = 10^3$.}
                \label{fig:de1}
        \end{subfigure}
        \qquad
        \begin{subfigure}{0.35\textwidth}
                \includegraphics[width=7.5cm,height=5cm]{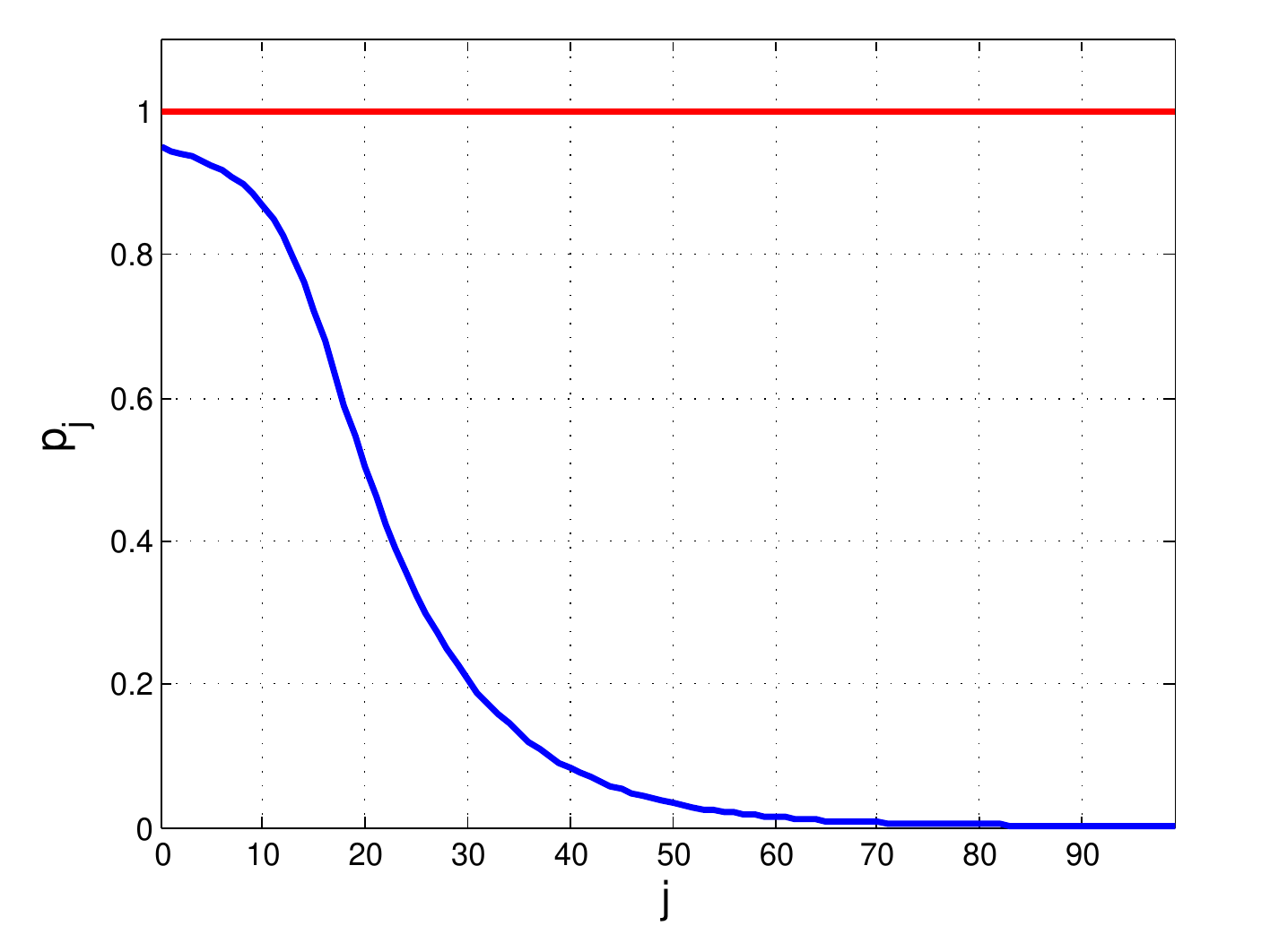}
                \caption{The evolution of $p_j$ after each iteration for parameters $K = 10^5$, $\epsilon=0.1$ and $D = 10^3$.}
                \label{fig:de2}
        \end{subfigure}
        \caption{Figure $(a)$ illustrates the density evolution equation similar to Figure \ref{fig:de1}. Figure $(b)$ illustrates the same phenomenon shows the evolution of $p_j$ versus the iteration, $j$.  Note that in this example, since $\epsilon = 0.1$ and we are operating very close to the capacity, $p_j$ gets very close to $0$ after around $90$ iterations. The reason is that the gap between the two curves in $(a)$ gets smaller as $\epsilon$ gets smaller. }\label{fig:density1}
\end{figure}

As shown in Lemma \ref{lem:de}, given that $f'(1)>1$, the density evolution has a fixed point at $1$. The other fixed point of the equation is approximately $p^* \simeq \lambda(e^{-\eta})$, which corresponds to the error floor of the algorithm. The intuition behind this is as follows. Assuming that one can get away from fixed point $1$ to $1 - \delta$ for some small $\delta > 0$, after a few iterations $p_j$ gets very close to $p^*$ as shown in Figures \ref{fig:density1} and \ref{fig:density2}. However, $p_j$ cannot get smaller than $p^*$, so $p^*$ is the error floor of the algorithm. In the following lemma, we show that for any arbitrarily small numbers $p^*$ and $\epsilon$, there exists a large enough constant $D(p^*,\epsilon)$ such that $f'(1) > 1$. This shows that with only $4M = 4K/(1-\epsilon) \simeq  4K(1+\epsilon)$ measurements, irregular PhaseCode algorithm can recover almost all the non-zero signal components. So given that the coloring procedure starts (the density evolution equation can be started from $1- \delta$), irregular PhaseCode is capacity-approaching. 

\begin{lemma}\label{lem:ir}
For any $p^* > 0$ and any $\epsilon > 0$, there exists a large enough constant $D(\epsilon,p^*)$ such that $M = K(1-\epsilon)^{-1}\simeq K(1+\epsilon)$ is the number of right nodes (bins), and $p_j$ converges to $p^*$ as $j$ goes to infinity, given that $p_1 = 1 - \delta$.
\end{lemma}

\begin{proof}
We show that if 
\begin{align} \label{eq:D}
D=\max\{(\frac{e}{1-\epsilon})^{2/\epsilon},(1+\frac{1}{p^*})^{1/(1-\epsilon)}\},
\end{align}
then, 
\begin{equation}\label{eq:stability}
f'(1)  = \eta e^{-\eta} \sum_{i \geq 1} \lambda_i (i-1) > 1,
\end{equation} 
and the error floor which is approximately $\lambda(e^{-\eta})$ is at most $p^*$; that is,
\begin{equation}\label{eq:error}
\sum_{i \geq 1} \lambda_i e^{-\eta(i-1)} \leq p^*.
\end{equation} 
This shows that in the density evolution equation, $p_j$ converges to $p^*$ as $j$ goes to infinity. This is illustrated in Figure \ref{fig:density2}. 

Recall that 
$$
\bar{d} = (\sum_{i=2}^D \frac{\lambda_i}{i})^{-1} = h(D-1) \frac{D}{D-1}.
$$
Thus, since $M = K/(1-\epsilon)$,
$$
\eta = \frac{K \bar{d}}{M} = h(D-1) \frac{D}{D-1} (1-\epsilon). 
$$ 
First, we show that \eqref{eq:error} in the following.

\begin{align*}
\sum_{i=2}^D \lambda_i e^{-\eta(i-1)} &= \frac{1}{h(D-1)} \sum_{i=2}^D \frac{1}{i-1}e^{-\eta(i-1)} \\
& \leq \frac{1}{h(D-1)} \sum_{i=1}^\infty e^{-\eta i} \\
& = \frac{e^{-\eta}}{h(D-1)(1-e^{-\eta})}. 
\end{align*}
It is enough to show that $h(D-1)(e^{\eta} - 1) \geq \frac{1}{p^*}$. We have
\begin{align*}
h(D-1)(e^{\eta} - 1) &\geq e^{\eta} - 1 \\
& \geq e^{\log(D).\frac{D}{D-1}(1-\epsilon)} -1 \\
& \geq D^{1- \epsilon} - 1 \\
& \geq \frac{1}{p^*},
\end{align*}
where the last inequality is due to \eqref{eq:D}.

Second, we show that \eqref{eq:stability} is satisfied in the following.

\begin{align}
\eta e^{-\eta} \sum_{i=2}^D \mu_i (i-1) &= \eta e^{-\eta} \frac{D-1}{h(D-1)} \\
&= D(1-\epsilon)e^{-h(D-1)\frac{D}{D-1}(1-\epsilon)} \\
& \geq D(1-\epsilon)e^{-(1+\log(D))\frac{D}{D-1}(1-\epsilon)} \\
& = \frac{1-\epsilon}{e} D^{\frac{\epsilon D - 1}{D-1}} \\ \label{eq:D1}
& \geq \frac{1-\epsilon}{e} D^{\epsilon/2} \\ \label{eq:D2}
& \geq 1,
\end{align}
where \eqref{eq:D1} is due to \eqref{eq:D} since $D \geq (\frac{e}{1-\epsilon})^{2/\epsilon} \geq \frac{2}{\epsilon}$ implies that $\frac{\epsilon D - 1}{D-1} \geq \frac{\epsilon}{2}$, and \eqref{eq:D2} is due to \eqref{eq:D}. This shows that $p_j, ~j\geq 1$ is a strictly decreasing sequence which is lower bounded by $p^*$. Thus, $p_j \to p^*$ as $j \to \infty$. This completes the proof.
\end{proof}

\begin{corollary}\label{cor:de}
Given that $p_1 = 1 - \delta$, for any $\epsilon_1 > 0$, there exists a constant $\ell(\epsilon_1)$ such that $p_\ell \leq p^* + \epsilon_1$. 
\end{corollary}
\begin{proof}
It is sufficient to show that after each iteration, $p_j$ decreases by a constant amount, that is a function of $\epsilon_1$. Let $q(\epsilon_1) = \arg \min_{p \in [p^* + \epsilon_1,1 - \delta]} p - f(p)$. Let $\gamma(\epsilon)= p_m - f(p_m)$, which is a strictly positive constant. Then, 
$$p_{j+1} - p_j = f(p_j) - p_j \leq - \gamma.$$ 
Therefore, it takes at most $\ell(\epsilon_1) = \frac{1 - p^*}{\gamma}<\frac{1}{\gamma}$ iterations  so that $ p_\ell \leq p^* + \epsilon_1$. 
\end{proof}

In the density evolution analysis so far, we have shown that the \emph{average} fraction of balls that remain uncolored will be arbitrarily close to the error floor after a fixed number of iterations, provided that the tree-like assumption is valid. It remains to show that the actual fraction of balls that are not in the giant component after $\ell$ iterations is highly concentrated around $p_\ell$. Since the maximum degree of left nodes are again a constant $D$, the exact procedure used in \cite{PLR14} to get a concentration bound can be also applied here.  Towards this end, first we show that a neighborhood of depth $\ell$ of a typical edge is a tree with high probability for a constant $\ell$ in Lemma \ref{lem:tree}.\footnote{The depth-$\ell$ neighborhood of edge $(v,c)$ is the subgraph of all the edges and nodes of paths having length less than or equal to $\ell$, that start from $v$ and the first edge of the path is not $(v, c)$.} Second, in Lemma \ref{lem:concentration}, we use the standard Doob's martingale argument \cite{RU01}, to show that the number of uncolored balls after $\ell$ iterations of the algorithm is highly concentrated around $K p_{\ell}$. We refer the readers to \cite{PLR14} for the proofs. 

Consider a directed edge from $\vec{e} = (v,c)$ from a left-node (ball) $v$ to a right-node (bin) $c$. Define the directed neighborhood of depth $\ell$ of $(\vec{e})$ as $\mathcal{N}_{\vec{e}}^{\ell}$, that is the subgraph of all the edges and nodes on paths having length less than or equal to $\ell$, that start from $v$ and the first edge of the path is not $\vec{e}$.

\begin{lemma}\label{lem:tree}
For a fixed $\ell^*$, $\mathcal{N}_{\vec{e}}^{2\ell^*}$  is a tree-like neighborhood with probability at least $1 - \mathcal{O}(\log(K)^{\ell^*}/K)$.
\end{lemma}

\begin{lemma}\label{lem:concentration}
Let $Z$ be the number of uncolored edges\footnote{An edge is colored if its corresponding ball is colored.} after $\ell$ iterations of the PhaseCode algorithm. Then, for any $\epsilon_2 > 0$, there exists a large enough $K$ and a constant $\beta$ such that 
\begin{align}\label{eq:ctcf}
|\mathbb{E}[Z] - K \bd p_\ell|& <  K \bd \epsilon_2 /2\\ \label{eq:mg}
\mathbb{P}(|Z - K\bd p_\ell| &> K\bd \epsilon) < 2e^{-\beta \epsilon_2^2 K^{1/(4\ell + 1)}},
\end{align} 
where $p_\ell$ is derived from the density evolution equation \eqref{eq:density}.
\end{lemma}

\subsection{Trigonometric-Based Modulation }\label{sec:trig}
In this section, we will explain the choice of the trig-modulation $T$, and how $T$ enables us to detect the non-zero signal components and recover them in magnitude and phase (up to a global phase). We will provide a brief explanation here, while referring the readers to \cite{PLR14} for thorough analysis. 

Define the length $4$ vector $y_i$ to be the measurement vector corresponding to the $i$-th row of matrix $H$ for $1 \leq i \leq M$. Then $y = [y_1^T, y_2^T, \ldots, y_M^T]^T$, where $y_i = [y_{i,1},y_{i,2},y_{i,3},y_{i,4}]^T$. Let $\om = \frac{\pi}{2n}$. Recall from  Equation \eqref{eq:T} that we design the measurement matrix $T = [t_{j\ell}]$ as follows. For all $\ell, ~ 1 \leq \ell \leq n$,
\begin{align*}
t_{1\ell} &= e^{\bi \om \ell}, \\
t_{2\ell} &= e^{-\bi \om \ell},  \\
t_{3\ell} &= 2 \cos(\om \ell), \\
t_{4\ell} &= e^{\bi \om' \ell}.
\end{align*}

This measurement matrix enables the algorithm to detect a left node (bin) that is connected to some colored balls and only one uncolored ball. Furthermore, the measurements enable the algorithm to find the index of uncolored ball, and recover the corresponding non-zero signal component in magnitude and phase relative to the known (colored) components. 

Consider a bin, let's say bin $i$, for which we know that it is connected to some colored balls. We want to check if bin $i$ is connected to exactly one other uncolored ball; i.e. one unknown non-zero component of $x$, say $x_\ell$, as in Figure \ref{fig:multiton}. We now describe a guess and check strategy to check our guess, and to find $\ell$ and $x_\ell$ if the guess is correct. If our guess is correct, we have access to measurements of the form:
\begin{align}\label{eq1}
y_{i,1} &= |a + e^{\bi \om \ell} x_\ell| ,\\ \label{eq2}
y_{i,2} &= |b + e^{-\bi \om \ell} x_\ell| , \\ \label{eq3}
y_{i,3} &= |c + 2\cos(\om \ell) x_\ell| , \\ \label{eq4}
y_{i,4} &= |d + e^{\bi \om' \ell} x_\ell|,
\end{align}
where complex numbers $a$, $b$, $c$ and $d$ are known values that depend on the values and locations of the known colored balls. We want to solve the first $3$ equations \eqref{eq1}-\eqref{eq3} to find $\ell$ and $x_\ell$, and use \eqref{eq4} to check if our guess is correct. We skip the algebra here, and refer the readers to \cite{PLR14} for details of how this can be done. We emphasize that one can solve the first 3 equations to find (at most) 4 possible solutions for $x_\ell$ and $\ell$ in closed-form. Whether one of those possible solutions is admissible or not can be checked using the $4$-th equation \eqref{eq4}. Since $\om '$ is a random phase, the probability that the check equation declares admissible solution while the solution is not valid is 0. 

\subsection{How to initialize the ball-coloring algorithm?}\label{sec:initial}
We have so far assumed that an arbitrarily small fraction $\delta > 0$ of the balls are colored in the initial phase. That is, a positive fraction of non-zero signal components are detected. In this subsection, we propose two methods for initializing the coloring algorithm. The first method is based on having an \emph{active} sensing stage. In active sensing, the measurements can be adaptive and are functions of the previous observations.  The second method assumes that the {\em locations}\footnote{Note that we assume only {\em location} knowledge of a tiny fraction of the active signal components, but no knowledge about
the magnitude or phase of these components.} of a small fraction of non-zero components of the signal are known. Note that this is the case for most practical scenarios. For example, almost all sparse images in the Fourier domain have non-zero low band components. In both methods, we make essential use of the deterministic measurements for non-sparse phase retrieval introduced in \cite{PLR14}. 

In the first method, as mentioned, we first try to find the location of a positive fraction of non-zero components. Consider the following $3 \epsilon_2 K$ measurements\footnote{Without loss of generality suppose that $\epsilon_2 K$ is an integer.} for some small $\epsilon_2 > 0$:
$$
\ty = |T_1 \otimes H_1 x|,
$$  
where 
\begin{equation}\label{eq:T1}
T_1 = \left( \begin{array}{cccc}
1 & 1 & \ldots & 1 \\
e^{\bi \om} & e^{\bi 2 \om} & \ldots & e^{\bi n \om} \\
\cos(\om) & \cos(2\om) & \ldots & \cos(n \om) 
\end{array}
\right) \in \mathbb{C}^{3 \times n},
\end{equation}
and $H_1 \in \{0,1\}^{\epsilon_2 K \times n}$ is a code matrix, and is constructed as follows. Each column of $H_1$ has exactly one entry that is equal to $1$, and the location of this entry is random. Equivalently, we consider a random bipartite graph model in which left nodes (which refer to signal components) are columns of $H_1$ and right nodes (which refer to measurements) are rows of $H_1$. Let $\ty = [\ty_1^T, \ty_2^T, \ldots, \ty_{M_1}^T]^T$, where $\ty_i = [\ty_{i,1},\ty_{i,2},\ty_{i,3}]^T$. Define a \emph{singleton} right node to be a right node that is connected to exactly 1 left nodes with non-zero signal component.  Now we show that if a right node is a singleton, the location and magnitude of the non-zero component can be found using ratio test. To this end, suppose that
\begin{align*}
\ty_{i,1} &= |x_\ell| \\
\ty_{i,2} &=  |e^{\bi \ell \om} x_\ell | \\
\ty_{i,3} &= |\cos(\ell \om) x_{\ell}|.
\end{align*} 
The fact that $\frac{\ty_{i,1}}{\ty_{i,2}} = 1$ reveals that the $i$-th right node is a singleton. The magnitude of the non-zero component is $\ty_{i,1}$. Furthermore, the location of the non-zero component can be found using $\ell = \cos^{-1}(\ty_{i,3}/x_{\ell})$. It is easy to see that the number of active components that can be found in location and magnitude approaches 
$$
K \Pr(\text{left node is connected to singleton} = K(1-\frac{1}{\epsilon_2 K})^K = Ke^{-1/\epsilon_2} = K\delta,
$$ 
where $\delta = e^{-1/\epsilon_2} > 0$ is a constant. In order to recover the phases of these non-zero components, we use deterministic measurements as follows. Let $\delta K = K_1$. Without loss of generality and for ease of notation assume that the detected non-zero components are $x_1,x_2,\ldots,x_{K_1}$. We use $2K_1 -2 $ measurements to get $|x_{1} + x_{\ell}|$ and $|x_{1} + \bi x_{\ell}|$ for $2 \leq \ell \leq K_1 -1$. Then, one can easily find the relative angle between $x_\ell$ and $x_1$. Therefore, a positive fraction $\delta$ of the active signal components can be found in location, phase and magnitude, and their corresponding balls can be colored. This ensures that the ball-coloring algorithm is initialized and it recovers almost all the signal components. Moreover, the number of extra measurements that we use in the active sensing stage is only $(3 \epsilon_2 + 2\delta)K$, and $(3 \epsilon_2 + 2\delta) > 0$ is a constant which can be made arbitrarily small. Thus, the total number of measurements that PhaseCode uses is still $4(1+\epsilon)K$ for arbitrarily small constant $\epsilon > 0$.    

The second method assumes that the location of a small fraction $\delta$ of active signal components is known, e.g. the low band components of an image in Fourier domain. Given this, one can use exactly the same procedure described in the first method to recover these components with $3\delta K - 2$ measurements. Again let $K_1 = \delta K$. Without loss of generality and for ease of notation assume that the known non-zero components are $x_1,x_2,\ldots,x_{K_1}$. Then the following measurements recover all the components using some simple algebra up to a global phase as described in \cite{PLR14}: $|x_{\ell}|, ~1 \leq \ell \leq K_1$, $|x_1 + x_{\ell}|, ~2 \leq \ell \leq K_1$, and 
$|x_1 + \bi x_{\ell}|, ~2 \leq \ell \leq K_1$.

\section{Two-Layer PhaseCode based on compressive sensing \\ and trig-measuremen-based phase-retrieval}

In this section, we propose a ``separation-principle" based two-layered approach to the compressive phase retrieval problem: a compressive-sensing outer layer and a deterministic trigonometric-measurement-based phase-retrieval inner layer.  A block diagram of our layered modular approach is shown in Figure \ref{fig:modular}.  As alluded to earlier,  a similar high-level approach has been independently proposed in recent work\cite{yapar}.  Although our works shares similarities in the high-level layering 
concepts,  the proposed algorithms are significantly different.  Principally, our measurement matrix is designed using sparse-graph codes, making our algorithm distinct from and more efficent than that of \cite{yapar}. 

\begin{figure}
\centering
    \includegraphics[width= \textwidth]{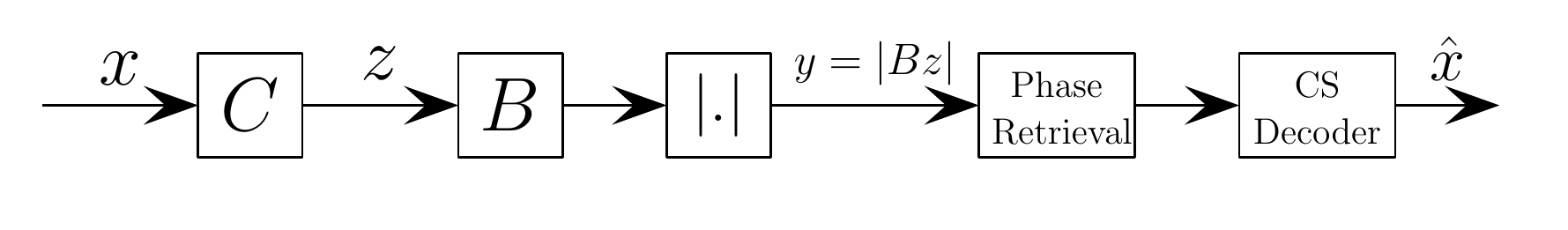}
  \caption{Block diagram of two-stage compressive phase retrieval. In this modular approach, the measurement matrix is the product of two matrices, $C$ and $B$. Matrix $C$ is an arbitrary compressive sensing (CS) matrix such that the CS decoder can recover signal $x$ from measurements $z = Cx$. matrix $B$ is a phase retrieval matrix that enables the algorithm to recover $z$ from the magnitude of linear measurements $y = |Bz|$.\label{fig:modular}} 
\end{figure}

Before we describe the technical details, let us overview the high-level proposed architecture.  In our two-layered design, the compressive-sensing layer  
operates as a linear phase-aware compressive measurement subsystem, while the trig-based phase-retrieval layer provides the desired abstraction between the targeted nonlinear phase-retrieval problem  and the induced linear compressive-sensing problem.  In our compressive sensing layer, we use the recent sparse-graph code based framework of \cite{Simon}, to design only $m_1 \simeq 2K$ measurements that guarantee perfect recovery for the compressive sensing problem, and in the phase-retrieval layer of the algorithm, we use $3m_1 -2$ deterministic measurements proposed in \cite{PLR14} that recover the phase of the $m_1$ measurements in phase and magnitude. Thus, we show that the signal can be reconstructed perfectly using only slightly more than $6K$ measurements.    Details follow.

\begin{proposition}{(Modular approach for compressive phase retrieval)} Let $C \in \mathbb{C}^{m_1 \times n}$ be a measurement matrix such that one can recover a $K$-sparse signal, $x\in \mathbb{C}^n$, from $m_1$ noiseless linear measurements, $z = Cx \in \mathbb{C}^{m_1}$, using a decoding algorithm of a certain computational complexity.
Assume that $z_1$ is non-zero. 
Then, one can recover a $K$-sparse signal $x$ from $(3m_1 -2)$ magnitude of linear measurements, $y = |BCx| = |Bz|$, using a modified decoding algorithm of the same computational complexity, where 
\begin{align*}
B &= [B_1^T,B_2^T,B_3^T] \in \mathbb{C}^{(3m_1 - 2) \times m_1}, \\
B_1 &= I_{m_1 \times m_1}, \\
B_2 &= \left( \begin{array}{cccccc}
1 & 1 & 0 & \ldots & ~ & 0\\
1 & 0 & 1 & 0 & \ldots & 0\\
1 & 0 & 0 & 1 & \ldots & 0 \\
~  &  ~ &  ~ &  \vdots  &   ~    & ~\\
1 & 0 & \ldots & ~ & 0 & 1  
\end{array}
\right) \in \mathbb{C}^{(m_1 - 1) \times m_1}, ~\text{and}\\
B_3 &= \left( \begin{array}{cccccc}
1 & \bi & 0 & \ldots & ~ & 0\\
1 & 0 & \bi & 0 & \ldots & 0\\
1 & 0 & 0 & \bi & \ldots & 0 \\
~  &  ~ &  ~ &  \vdots  &   ~    & ~\\
1 & 0 & \ldots & ~ & 0 & \bi  
\end{array}
\right) \in \mathbb{C}^{(m_1 - 1) \times m_1}.
\end{align*} 
\label{prop:trick}
\end{proposition}

\begin{proof}
Since $|z_1| \neq 0$, one can retrieve the phase of $z_\ell$ for all $\ell$ up to a global phase using measurements $y_1 = |z_1|$, $y_{\ell} = |z_\ell|$, $y_{m_1 + \ell - 1} = |z_1 + z_\ell|$ and $y_{2m_1 + \ell - 2} = |z_1 + \bi z_\ell|$ for $\ell \geq 2$ as follows. Let $z_\ell = |z_\ell| e^{\bi \phi}$. Then, by the cosine law we have
$$  
\cos(\phi) = \frac{|z_1 + z_\ell|^2 - |z_1|^2 - |z_\ell|^2}{2|z_1 z_\ell|} = \alpha.
$$
Assuming that $\cos^{-1}(\alpha)$ for  $-1 \leq \alpha \leq 1$ returns values between $0$ and $\pi$, we have $\phi = \pm \cos^{-1}(\alpha)$. Thus, we need to resolve the ambiguity of the sign. This ambiguity can be resolved using the other measurement: $|z_1 + \bi z_\ell|$.

Thus, $z$ is recovered up to a global phase, and one can decode $x$ up to a global phase from it using the decoding algorithm for linear measurements. 
\end{proof}

Note that the proposed approach is highly flexible due to its modularity: one can choose any compressive measurement matrix $C$ that possesses suitable properties that are specific to applications. 

Recently in \cite{Simon}, it is shown that a $K$-sparse signal, $x\in \mathbb{C}^n$, can be perfectly recovered with high probability using only $m_1 = 2K(1+\epsilon)$ noiseless linear measurements. The measurement matrix is designed based on irregular left-degree sparse-graph codes. Using this compressive measurement scheme alongside Proposition \ref{prop:trick}, we provide an efficient compressive phase retrieval scheme.
\begin{corollary}
For arbitrarily small $\epsilon > 0$, with high probability, one can \emph{perfectly} recover a $K$-sparse signal with magnitude of $6K(1+\epsilon)-2$ linear measurements and optimal complexity $\mathcal{O}(K)$.
\end{corollary} 

\section{Conclusion}

\begin{figure}[h]
\centering
    \includegraphics[width= 0.7\textwidth]{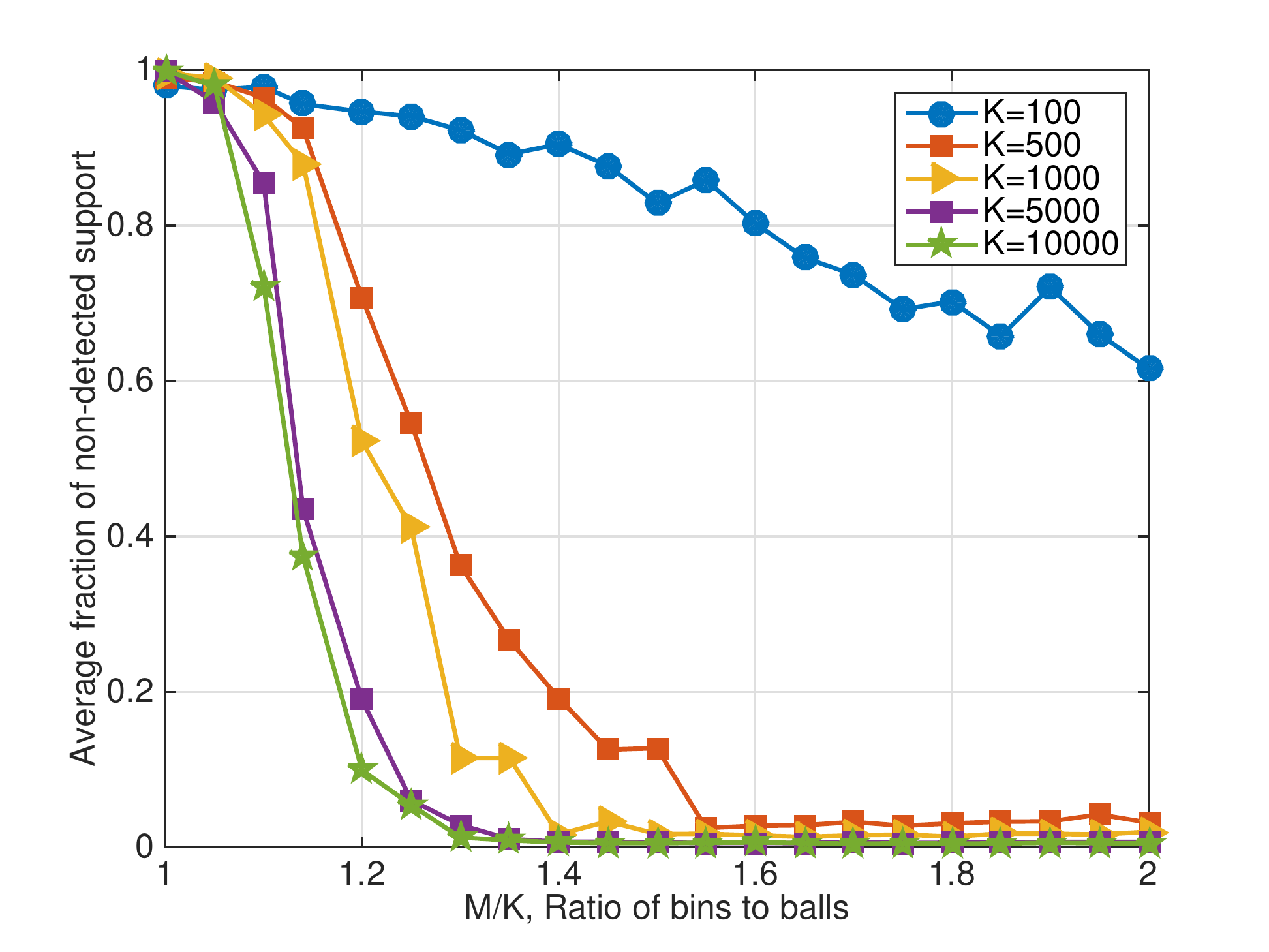}
  \caption{Performance of irregular PhaseCode {\em without initialization phase}:  We simulate the ball-coloring algorithm without any initialization phase and varying the ratio of the number of right nodes in the graph (bins) to the number of non-zero signal components (balls), $M/K$. We consider several values of $K$ from $100$ to $10000$, and fix the length of the signal as $n = 20000$. 
The plotted simulation results show a clear trend: as $K$ increases, the sample complexity required to decode successfully decreases and it approaches the capacity that is $M/K = 1$. For instance, when $K = 10000$, the coloring algorithm successfully decodes all symbols but a small fraction ($\epsilon \leq 0.005$) when $M/K = 1.3$.} 
\end{figure}

In this paper, we have considered the problem of recovering a $K$-sparse complex signal $x \in \mathbb{C}^n$ from $m$ intensity measurements of the form $|a_i x|, ~1 \leq i \leq m$, where $a_i$ is a measurement row vector. It is well-known that the minimum number of measurements to perfectly recover the signal is $4K - o(K)$ for compressive phase retrieval. We have proposed an irregular-left-degree PhaseCode algorithm that can recover the signal almost perfectly using only $4K(1+\epsilon)$ measurements for arbitrarily small $\epsilon$ with high probability as $K$ gets large under some mild assumptions.  In order to initialize our algorithm, we rely on either an active sensing phase with $\epsilon K$ measurements, or we assume that the locations of an arbitrarily small fraction of the non-zero components are known, which is very reasonable in many applications of interest, e.g. in the case of natural objects like images, we know that low-frequency bands are active.

We have also proposed another variant of PhaseCode algorithm that is based on a separation architecture involving compressive sensing using sparse-graph codes. The measurement matrix of this algorithm is a concatenation of deterministic measurements that are used to recover the phase of the measurements required for compressive sensing, and measurements based on sparse-graph codes for compressive sensing which are now phase-aware. We have shown that the signal can be reconstructed perfectly using only slightly more than $6K$ measurements.    

We conclude the paper by emphasizing an empirical observation that is useful for practical applications. In Section \ref{sec:initial}, we explained how one can initialize the coloring algorithm by means of feedback or prior knowledge about a few non-zero components of the signal. In our previous work \cite{PLR14}, we showed that a linear-size component of colored balls can be formed with around $14K$ measurements (which is far from the fundamental limit $4K$) in the regular construction without using feedback or any prior knowledge about the signal. However, via simulations, we empirically observe that our coloring algorithm can be used even without the initialization stage for irregular PhaseCode. We simulate irregular PhaseCode without any initialization stage with the following parameters. We vary $K$ from $100$ to $10000$, and fix the length of signal as $n = 20000$. We observe that even without the initialization stage, the coloring algorithm successfully recovers the signal with a sample complexity that is close to the fundamental limit. For example, when $K = 10000$, the coloring algorithm successfully decodes the signal with only $4 \times 1.3K = 5.2K$ measurements. An interesting future direction is to theoretically investigate whether irregular PhaseCode without initializing is also capacity-approaching.  Finally, in this work, we have focused on the noiseless compressive phase-retrieval problem, predominantly to elucidate our sparse-graph based capacity-approaching construction.  Of course, in most practical settings, one needs to consider noisy observations.  Our proposed framework can indeed be extended to be robust to noise, and this will be part of our future dissemination.


\end{document}